\documentclass{eptcs}
\usepackage[latin1]{inputenc}
\usepackage{mathpartir}
\usepackage{latexsym}
\usepackage{amssymb}
\usepackage{amsmath}
\usepackage{macros}
\usepackage{color}
\usepackage{graphics}
\usepackage{version}
\usepackage{url}

\providecommand{\event}{B. Klin, P. Soboci\'nski (Eds.): \\
6th Workshop on Structural Operational Semantics (SOS'09)}
\providecommand{\volume}{19}
\providecommand{\anno}{2010}
\providecommand{\firstpage}{17}
\providecommand{\eid}{2}
\usepackage{breakurl}

\usepackage{smacros}

\usepackage{amsthm}
\newtheorem{definition}{Definition}
\newtheorem{lemma}[definition]{Lemma}
\newtheorem{theorem}[definition]{Theorem}
\newtheorem{corollary}[definition]{Corollary}

\theoremstyle{definition}

\title{A Fully Abstract Symbolic Semantics for Psi-Calculi}
\author{Magnus Johansson \and Bj{\"o}rn Victor \and Joachim Parrow}
\def\titlerunning{A Fully Abstract Symbolic Semantics for Psi-Calculi}
\def\authorrunning{Magnus Johansson, Bj{\"o}rn Victor, Joachim Parrow}
\date{\today}
\begin{document}
\maketitle

\begin{abstract}
We present a symbolic transition system and bisimulation equivalence
for psi-calculi, and show that it is fully
abstract with respect to bisimulation congruence in the non-symbolic semantics.

A psi-calculus is an extension of the pi-calculus with nominal data
types for data structures and for logical assertions representing
facts about data. These can be transmitted between processes and their
names can be statically scoped using the standard pi-calculus
mechanism to allow for scope migrations.
Psi-calculi can be more general than other proposed extensions of the
pi-calculus such as the applied pi-calculus, the spi-calculus, the
fusion calculus, or the concurrent constraint pi-calculus.

Symbolic semantics are necessary for an efficient implementation of
the calculus in automated tools exploring state spaces, and the full abstraction property means the
semantics of a process does not change from the original.
\end{abstract}

\section{Introduction}\label{sec:introduction}

A multitude of extensions of the pi-calculus have been defined,
allowing higher-level data structures and operations on them to be
used as primitives when modelling applications.  Ranging from
integers, lists, or booleans to encryption/decryption or hash
functions, the extensions increase the applicability of the basic calculus.
In order to implement automated tools for analysis and verification 
using state space exploration (e.g. bisimilarity or model checking),
each extended calculus needs a symbolic semantics, where the state space of
agents is reduced to a manageable size -- the non-symbolic semantics
typically generates infinite state spaces even for very simple agents. 

The extensions thus require added efforts both in developing the
theory of the calculus for each variant, and in constructing
specialised symbolic semantics for them.
As the complexity of the extensions increases, producing correct
results in these areas can be very hard. 
For example the labelled semantics of applied
pi-calculus~\cite{abadi.fournet:mobile-values} and of
CC-Pi~\cite{buscemi.montanari:open-bisimulation} have both turned out
to be non-compositional; another example is the rather complex  
bisimulations which have been developed for the
spi-calculus~\cite{abadi.gordon:calculus-cryptographic} (see
\cite{borgstrom.nestmann:bisimulations-spi} for an overview of
non-symbolic bisimulations, or
\cite{borgstroem.briais.ea:symbolic-bisimulation,briais:theory-tool,borgstroem:equivalences-calculi}
for symbolic ones).

The psi-calculi~\cite{bengtson.johansson.ea:psi-calculi} improve the situation: a single framework allows a range of
specialised calculi to be formulated with a lean and compositional labelled
semantics: with the parameters appropriately instantiated, the
resulting calculus can be used to model applications such as
cryptographic protocols %, object-oriented programming, 
and concurrent
constraints, but also more advanced scenarios with polyadic
synchronization or % frequency hopping, multicast communication, or
higher-order data and logics.
The expressiveness and modelling convenience of psi-calculi exceeds that
of earlier pi-calculus extensions, while the purity of the semantics
is on par with the original pi-calculus. Its meta-theory has been
proved mechanically using the theorem prover Isabelle~\cite{bengtson.parrow:psi-calculi-isabelle}.

In this paper we develop a symbolic semantics for psi-calculi,
admitting large parts of this range of calculi to be verified more efficiently.
We define a symbolic version of labelled bisimulation equivalence, and
show that it is fully abstract with respect to bisimulation congruence
in the original semantics.  This means that our new symbolic semantics
does not change which processes are considered equivalent.

%%%%

A symbolic semantics abstracts the values received in an input
action. Instead of a possibly infinite branching of concrete values, a
single name is used to represent them all.  When the received values are
used in conditional constructions (e.g. if-then-else) or as
communication channels, we do not know their precise value, but need to
record the constraints which must be satisfied for a resulting
transition to be valid. 

A (non-symbolic) psi-calculus transition has the form $\Psi\frames P\gta P'$, with
the intuition that $P$ can perform $\alpha$ leading to $P'$ in an
environment that asserts $\Psi$.  
For example, if $P$ can
do an $\alpha$ to $P'$ then $\caseonly{\mathrm{prime}(x)}:{P}$ can make an
$\alpha$-transition
to $P'$ if we can deduce $\mathrm{prime}(x)$ from the environment, 
% if $\Psi\vdash\mathrm{prime}(x)$,
e.g. 
\[ \{x=3\}\frames \trans{\caseonly{\mathrm{prime}(x)}:{P}}{\alpha}{P'}. \]
In the symbolic semantics where we may not have the precise value of
$x$, we instead decorate the transition with its
requirement, so 
\[
\Psi\frames\transs{\caseonly{\mathrm{prime}(x)}:
P}{\alpha}{C \wedge \constr{\Psi\vdash\mathrm{prime}(x)}}{P'} \text{
(for any $\Psi$)}
\]
where $C$ is the requirement for $P$ to do an $\alpha$ to $P'$ in the
environment $\Psi$.
Constraints also arise from communication between parallel agents,
where, in the symbolic case, the precise channels may not be known;
instead we allow communication over symbolic representations of channels and
record the requirement in a transition constraint. As an example
% from the pi-calculus, 
consider 
$a(x)\sdot a(y)\sdot (\opm{x}{x}\sdot P\parop y(z)\sdot
Q)$ which after its initial inputs only has symbolic values of $x$ and
$y$. The resulting agent has the symbolic transition
$\opm{x}{x}\sdot P\parop y(z)\sdot Q\sgto{\tau}{\constr{\Psi\vdash x\sch y}}P\parop Q\lsubst{x}{z}$ where
$x\sch y$ means that $x$ and $y$ represent the same channel, but might
not have a $\tau$ transition in the non-symbolic semantics.

Communication channels in psi-calculi may be structured data terms,
not only names. This leads to a new source of possibly infinite
branching: a subject in a prefix may be rewritten to another
equivalent term before it is used in a transition. E.g., when $\mathsf{first}(x,y)$
and $x$ represent the same channel,
$P=\op{\mathsf{first}(a,b)}c\sdot P'\gt{\op{a}c}P'$, but also $P\gt{\op{\mathsf{first}(a,c)}c}P'$,
etc. % This is different from e.g. \api{}, but 
The possibility of using
structured channels gives significant expressive power
(see~\cite{bengtson.johansson.ea:psi-calculi}). Our symbolic semantics
abstracts the equivalent forms of channel subject by using a fresh
name as subject, and adds a suitable constraint to the transition
label (see Section~\ref{sec:symbolic}).

\subsection{Comparison to related work}
Symbolic bisimulations for process calculi have a long history.
Our work is to a large extent based on the pioneering work by Hennessy
and Lin~\cite{hennessy.lin:symbolic-bisimulations} for value-passing
CCS, later specialised for the \pic{} by Boreale and
De~Nicola~\cite{boreale.de-nicola:symbolic-semantics} and independently
by Lin~\cite{lin:symbolic-transition,lin:computing-bisimulations}.
While~\cite{hennessy.lin:symbolic-bisimulations} is
parameterised by general boolean expressions on an underlying data
signature it does not handle names and mobility; on the other hand
\cite{boreale.de-nicola:symbolic-semantics,lin:symbolic-transition,lin:computing-bisimulations}
handle \emph{only} names and no other data structures.
The number of (direct or indirect) follow-up works to these is huge,
with applications ranging from pi-calculus to constraint programming;
here we focus on the relation to the ones for \api{} and spi-calculus.

The existing tools for calculi based on the \api{}
(e.g.~\cite{abadi.blanchet:analyzing-security,blanchet:efficient-cryptographic,blanchet.abadi.ea:automated-verification}),
are not fully abstract wrt bisimulation.  A symbolic
semantics and bisimulation for \api{} has been defined in
\cite{DBLP:conf/fsttcs/DelauneKR07}, but it is not complete.
Additionally, the labelled (non-symbolic) bisimulation of \api{} is
not compositional (see~\cite{bengtson.johansson.ea:psi-calculi}).
The situation for
the spi-calculus %~\cite{abadi.gordon:calculus-cryptographic} 
is better: fully abstract symbolic bisimulation for hedged
bisimulation has been defined in
\cite{borgstroem:equivalences-calculi}, and for open hedged
bisimulation (a finer equivalence) in \cite{briais:theory-tool}.  
According to those authors, neither is directly mechanizable.
The only symbolic
bisimulation which to our knowledge has been implemented in a tool is
not fully abstract~\cite{borgstroem.briais.ea:symbolic-bisimulation}.

It can be argued~\cite{borgstroem.briais.ea:symbolic-bisimulation}
that incompleteness is not a problem when verifying authentication and
secrecy properties of security protocols, which appears to have been the main
application of the \api{} so far.  
When going beyond security analysis we claim (based on experience from the Mobility
Workbench~\cite{victor.moller:mobility-workbench}) that
completeness is very important: when analysing
agents with huge state spaces, a positive result (the agents are
equivalent) may be more difficult to achieve than a negative result
(the agents differ).  However, such a negative result can only be
trusted if the analysis is fully abstract.

Our symbolic semantics is relatively simple, compared to the ones presented for \api{} or spi-calculus. 
In relation to the former, we are helped significantly by the absence
of structural equivalence rules, which in \api{} are rather
complex. In~\cite{DBLP:conf/fsttcs/DelauneKR07} an intermediate
semantics is used to overcome the complexity. In contrast we can directly
relate the original and symbolic semantics.
In relation to the symbolic semantics for spi-calculus, our semantics
has a straight-forward treatment of scope opening due to the simpler
psi-calculi semantics. In addition, the complexities of spi-calculus
bisimulations  are necessarily inherited by the symbolic semantics,
introducing e.g. explicit environment knowledge representations with
timestamps on messages and variables.  In psi-calculi, bisimulation is
much simpler and the symbolic counterpart is not significantly more
complex than the one for value-passing CCS.

In the light of these complications, the relevance of
precise encodings of the \api{} or spi-calculus as psi-calculi,
or comparing the resulting bisimulation equivalences is questionable.
Our interest is
in handling and analysing the same type of applications, and also
the more advanced kinds of applications mentioned in the beginning of this section. 

\paragraph{Disposition.}
In the next section we review the basic definitions of syntax,
semantics, and bisimulation of psi-calculi. Section~\ref{sec:symbolic}
presents the symbolic semantics and bisimulation, while
Section~\ref{sec:examples} illustrates the concrete and symbolic
transitions and bisimulations by examples. In
Section~\ref{sec:results} we show our main results: the correspondence
between concrete and symbolic transitions and
bisimulations. Section~\ref{sec:conclusion} concludes, and presents
plans and ideas for future work.

\section{Psi-calculi}\label{sec:psi}

This section is a brief recapitulation of psi-calculi and nominal data types; for a more extensive treatment including motivations and examples see~\cite{bengtson.johansson.ea:psi-calculi}.

\subsection{Nominal data types}

We assume a countably infinite set of atomic {\em names} $\nameset$ ranged over by $a,b,\ldots,x,y,z$. Intuitively, names will represent the symbols that can be statically scoped, and also represent symbols acting as variables in the sense that they can be subjected to substitution. 
A {\em nominal set}~\cite{PittsAM:nomlfo-jv,Gabbay01anew} is a set equipped with {\em name swapping} functions written $(a\;b)$, for any names $a,b$. An intuition is that for any member $X$ it holds that $(a\;b)\cdot X$ is $X$ with  $a$ replaced by $b$ and  $b$ replaced by $a$. One main point of this is that  even though we have not defined any particular syntax we can define what it means for a name to ``occur" in an element: it is simply that it can be affected by swappings. 
% Formally, the {\em support} of $X$, written $\n(X)$, is the least set of names $A$ such that $(a\;b)\cdot X = X$ for all $a,b$ not in $A$.
The names occurring in this way  in an element $X$ constitute the  {\em support} of $X$, written $\n(X)$.

  We write $a \freshin X$, pronounced ``$a$ is fresh for $X$", for $a \not\in \n(X)$.  If $A$ is a set of names we write $A\freshin X$ to mean $\forall a \in A \;.\; a \freshin X$.
We require all elements to have finite support, i.e., $\n(X)$ is finite for all
$X$.

A function $f$ on nominal sets is {\em equivariant} if $(a\;b)\cdot f(X)
= f((a\;b)\cdot X)$ holds for all $X,a,b$, and similarly for functions and
relations of any
arity. Intuitively, this means that all
names are treated equally.

A {\em nominal data type} is just a nominal set together with a set of functions on it. In particular we require a substitution function~\cite{bengtson.parrow:psi-calculi-isabelle}, which intuitively substitutes elements  for  names. If $X$ is an element of a data type, $\tilde{a}$ is a sequence of names without duplicates and $\tilde{Y}$ is an equally long sequence of elements, the {\em substitution}
$X\lsubst{\tilde{Y}}{\tilde{a}}$ is an element of the same data type as $X$.

\subsection{Agents}

\newcommand{\terms}{{\rm\bf T}}
\newcommand{\conditions}{{\rm\bf C}}
\newcommand{\assertions}{{\rm\bf A}}

A psi-calculus is defined by instantiating three nominal data types and four operators:
\begin{definition}[Psi-calculus parameters]
\label{def:parameters}
A psi-calculus requires the three (not necessarily disjoint) nominal data types:
\[\begin{array}{ll}
\terms & \mbox{the (data) terms, ranged over by $M,N$} \\
\conditions  & \mbox{the conditions, ranged over by $\varphi$}\\
\assertions & \mbox{the assertions, ranged over by $\Psi$}
\end{array}\]
and the four equivariant operators:
\[\begin{array}{ll}
\sch:  \terms \times \terms \to \conditions & \mbox{Channel Equivalence} \\
\ftimes: \assertions \times \assertions \to \assertions& \mbox{Composition} \\
\emptyframe: \assertions& \mbox{Unit} \\
\vdash\,\subseteq \assertions \times \conditions & \mbox{Entailment}
\end{array}
\]
\end{definition}
The binary functions above will be written in infix. Thus, if $M$ and $N$ are terms then $M \sch N$ is a condition, pronounced ``$M$ and $N$ are channel equivalent" and if $\Psi$ and $\Psi'$ are assertions then so is $\Psi \ftimes \Psi'$. Also we write $\Psi \vdash \varphi$, ``$\Psi$ entails $\varphi$", for $(\Psi, \varphi) \in \;\vdash$.

The data terms are used to represent all kinds of data, including communication
channels. Conditions are used as guards in agents, and $M \sch N$ is a
particular condition saying that $M$ and $N$ represent the
same channel.
The assertions will be used to declare information necessary to resolve the
conditions. Assertions can be contained in agents and thus represent information
postulated by that agent; they can contain names and thereby be syntactically
scoped and thus represent information known only to the agents within that
scope. 
The intuition of entailment is that $\Psi \vdash \varphi$ means that given the
information in $\Psi$, it is possible to infer $\varphi$.
 We say that two assertions are equivalent if they entail the same conditions:
\begin{definition}[Assertion equivalence]
\label{def:assEq}
Two assertions are {\em equivalent}, written $\Psi \sequivalent \Psi'$, if for all $\varphi$ we have that $\Psi \vdash \varphi \Leftrightarrow\Psi' \vdash \varphi$.
\end{definition}

A psi-calculus is formed by instantiating the nominal data types and morphisms so that the following requisites are satisfied:
\begin{definition}[Requisites on valid psi-calculus parameters]
\label{def:entailmentrelation}
\
\begin{mathpar}
\begin{array}{ll}
\mbox{Channel Symmetry:} & \Psi \vdash M \sch N \; \Longrightarrow\; \Psi \vdash N \sch M \\
\mbox{Channel Transitivity:} & \Psi \vdash M \sch N \; \land \; \Psi \vdash N \sch L\\
& \quad \quad \;\Longrightarrow\; \Psi \vdash M \sch L\\ 
\\
\mbox{Weakening:} & \Psi \vdash \varphi \; \Longrightarrow \; \Psi \ftimes \Psi' \vdash \varphi \\
\\ 

\mbox{Composition:} & \Psi \sequivalent \Psi'  \;\Longrightarrow\; \Psi \ftimes \Psi'' \sequivalent \Psi' \ftimes \Psi''\\
\mbox{Identity:} & \Psi \ftimes \emptyframe \sequivalent \Psi \\
\mbox{Associativity:}& (\Psi \ftimes \Psi') \ftimes \Psi'' \sequivalent \Psi \ftimes (\Psi' \ftimes \Psi'')\\
\mbox{Commutativity:}&  \Psi \ftimes \Psi' \sequivalent \Psi' \ftimes \Psi \\

\end{array}
\end{mathpar}
\end{definition}
\noindent

Our requisites on a psi-calculus are that the channel equivalence is a partial
equivalence relation, that $\ftimes$ preserves equivalence, and that the
equivalence classes of assertions form an abelian monoid. 
We do not require that channel equivalence is reflexive. There may be
terms $M$ such that $M\sch M$ does not hold. By transitivity and
symmetry then $M\sch N$ holds for no $N$, which means that $M$ cannot
be used as a channel at all. In this way we accommodate data
structures which cannot be used as channels.
% We do not require
% that channel equivalence is reflexive and thus allow terms which may not be used
% as channels. 
The requisite of weakening (which is not present in
\cite{bengtson.johansson.ea:psi-calculi}) excludes some non-monotonic logics; it
simplifies our proofs in the present paper although we do not know if it is
absolutely necessary. It is only used in one place in the
proof of Theorem~\ref{theorem:soundness}.

In the following $\tilde{a}$ means a finite (possibly empty) sequence of names, $a_1,\ldots,a_n$. The empty sequence is written $\epsilon$ and the concatenation of $\tilde{a}$ and $\tilde{b}$ is written $\tilde{a} \tilde{b}$.
When occurring as an operand of a set operator, $\tilde{a}$ means the corresponding set of names $\{a_1,\ldots, a_n\}$. We also use sequences of terms, conditions, assertions etc. in the same way.

A {\em frame} can intuitively be thought of as an assertion with local names:
\begin{definition}[Frame]
\label{def:frame}
A {\em frame} $F$ is of the form $\framepair{\frnames{}}{\frass{}}$ where
$\frnames{}$ is a sequence of names considered bound in
the assertion 
$\frass{}$. We use $F,G$ to range over frames. \footnote{In some presentations
frames have been written just as pairs $\langle \frnames{},\frass{} \rangle$.
The notation in this paper better conveys the idea that the names bind into the
assertion, at the slight risk of confusing frames with agents. Formally, we
establish frames and agents as separate types, although a valid intuition is to
regard a frame as a special kind of agent, containg only scoping and assertions.
This is the view taken in~\cite{abadi.fournet:mobile-values}. }

\end{definition}
Name swapping on a frame $F =
\framepair{\frnames{}}{\frass{}}$ just distributes to its two components. We
identify
alpha equivalent frames, so $\n(F) = \n(\frass{}) - \n(\frnames{})$. 
We overload $\emptyframe$ to also mean the least informative frame
$\framepair{\epsilon}{\emptyframe}$ and $\ftimes$ to mean composition on frames
defined by $\framepair{\frnames{1}}{\frass{1}} \ftimes
\framepair{\frnames{2}}{\frass{2}} = 
\framepair{\frnames{1} \frnames{2}}{\frass{1} \ftimes \frass{2}}$ where
$\frnames{1}$ $\freshin$ $\frnames{2},\frass{2}$ and vice versa. We write
$(\nu c)(\framepair{\frnames{}}{\frass{}})$
for $\framepair{c\frnames{}}{\frass{}}$, and when there is
no risk of confusing a frame with an assertion we write $\frass{}$ for
$\framepair{\epsilon}{\frass{}}$.

\begin{definition}[Equivalence of frames]\label{def:frame-equivalence}
We define $F \vdash \varphi$ to mean that there exists an alpha variant
$\framepair{\frnames{}}{\frass{}}$ of $F$ such that  $\frnames{}
\freshin \varphi$ and $\frass{} \vdash \varphi$. We also define 
$F\sequivalent G$ to mean that for all $\varphi$ it holds that $ F \vdash
\varphi$ iff $ G \vdash \varphi$.
\end{definition}
Intuitively a condition is entailed by a frame if it is entailed by the
assertion and does not contain any names bound by the frame. 
Two frames are equivalent if they entail the same conditions.

\begin{definition}[Psi-calculus agents]\label{def:agents}
Given valid psi-calculus parameters as in Definitions~\ref{def:parameters} and~\ref{def:entailmentrelation}, the psi-calculus {\em agents}, ranged over by $P,Q,\ldots$,  are of the following forms.
{\rm
\[
\begin{array}{ll}

%\nil                          & \mbox{Nil} \\
\out{M}N .P                   & \mbox{Output} \\
\inprefix{M}{\ve{x}}{x}.P          & \mbox{Input}\\
%\mbox{\rm\bf case}\; \varphi_1. P_1 \; \casesep \cdots\casesep\;\varphi_n . P_n  
\caseonly{\ci{\varphi_1}{P_1}\casesep\cdots\casesep\ci{\varphi_n}{P_n}}
&\mbox{Case} \\
%\caseonly{\ci{\ve{\varphi}}{\ve{P}}} & \mbox{Case} \\
(\nu a)P                      & \mbox{Restriction}\\
P \pll Q                      & \mbox{Parallel}\\
! P                           & \mbox{Replication} \\
\pass{\Psi}                        & \mbox{Assertion}\\
\end{array}\]
}

\noindent
In the Input $\inprefix{M}{\ve{x}}{x}.P$,  $x$ binds its occurrences in $P$. 
Restriction binds $a$ in $P$. An assertion is {\em guarded} if it is a subterm
of an Input or Output. In a replication $!P$ there may be no unguarded
assertions in $P$.
% , and in $\caseonly{\ci{\varphi_1}{P_1}\casesep\cdots\casesep\ci{\varphi_n}{P_n}}$ there may be no unguarded assertion in any $P_i$
\end{definition}
In the Output and Input forms $M$ is called the subject and
$N$ and $x$ the objects, respectively.
Output and Input  are similar to those in the pi-calculus, but
arbitrary terms can function as both subjects and objects.
Note that differently from~\cite{bengtson.johansson.ea:psi-calculi},
for simplicity the input is not pattern matching (see
Section~\ref{sec:conclusion} for a discussion).
The {\bf case} construct works by performing the action of any
$P_i$ for which the corresponding $\varphi_i$ is true. So it
embodies both an {\bf if} (if there is only one branch) and an internal
nondeterministic choice (if the conditions are overlapping).

Some notational conventions: We define the agent $\nil$ as $\pass{\emptyframe}$.
%We use $\tau.P$ as a shorthand for $(\nu b)(\overline{b}b.\nil\parop\underline{b}(\lambda b)b.P)$ for some $b\freshin P$.
The  construct 
$\caseonly{\ci{\varphi_1}{P_1}\casesep\cdots\casesep\ci{\varphi_n}{P_n}}$ is
sometimes written as
\mbox{\rm $\caseonly{\ci{\ve{\varphi}}{\ve{P}}}$}, or if $n=1$ as 
$\ifthen{\varphi_1}{P_1}$.
%In psi-calculi where a condition $\top$ exists such that $\Psi \vdash \top$ for all $\Psi$
% we write $P+Q$ to mean $\caseonly{\ci{\top}{P}\casesep\ci{\top}{Q}}$.
The input subject is underlined to facilitate parsing of complicated
expressions; in simple cases we often conform to a more traditional notation and
omit the underline.

Formally, we define name swapping on agents by distributing it over all constructors, and substitution on agents by distributing it and avoiding captures by binders through alpha-conversion in the usual way. We identify alpha-equivalent agents; in that way we get a nominal data type of agents where
the support $\n(P)$ of $P$ is the union of
the supports of the components of $P$, removing the names bound by Input and
$\nu$, and corresponds to the names with a free occurrence in $P$.

\begin{definition}[Frame of an agent]
The {\em frame $\fr{P}$ of an agent} P is defined inductively as follows:
\[\begin{array}{l}
\fr{\inprefix{M}{\ve{x}}{x}.P} = \fr{\out{M}N.P} = \\
\qquad \qquad \fr{\caseonly{\ci{\ve{\varphi}}{\ve{P}}}} =
\fr{!P} = \emptyframe \\
\fr{\pass{\Psi}} = \framepair{\epsilon}{\Psi} \\
\fr{P \pll Q} = \fr{P}\ \ftimes\ \fr{Q}\\
\fr{\res{b}P} = (\nu b)\fr{P}  \\
% \fr{!P} = \fr{P}
 \end{array}\]
\end{definition}

%\newpage
\subsection{Operational semantics}

The presentation of psi-calculi in~\cite{bengtson.johansson.ea:psi-calculi} gives a semantics of an early kind, where input actions are of kind $\inn{M}N$. 
Here we give an operational semantics of the late kind, meaning that
the labels of input transitions contain variables for the object to be
received. With this kind of semantics it is easier to establish a
relation to the symbolic semantics. We also establish precisely how it
relates to the original.

\begin{definition}[Actions]

The {\em actions} ranged over by $\alpha, \beta$ are of the following three kinds:
$\out{M}{(\nu \tilde{a})N}$ (Output), $\inlabel{M}{\ve{x}}{x}$ (Input), and
$\tau$ (Silent).

% \[\begin{array}{lllll}

% \out{M}{(\nu \tilde{a})N} & \mbox{Output} \\
% \lin{M}{x}{x}                  & \mbox{Input}  \\
% \tau                      & \mbox{Silent}
% \end{array}\]

\end{definition}

For actions we refer to $M$ as the {\em subject} and $N$ and $x$ as the {\em
objects}. We let $\subj{\bout{\ve{a}}{M}{N}} =
\subj{\inlabel{M}{}{x}} = M$. We define 
$\bn{\out{M}{(\nu \tilde{a})N}} = \tilde{a}$, 
$\bn{\inlabel{M}{\ve{x}}{x}} = \{x\}$, 
and $\bn{\tau}=\emptyset$. We also define $\n(\tau)=\emptyset$ and
$\n(\alpha) = \n(N) \cup \n(M)$ if $\alpha$ is an output or input.
As in the pi-calculus, the output $\out{M}{(\nu \tilde{a})N}$ represents an
action sending $N$ along $M$ and opening the scopes of the names $\tilde{a}$. 
Note in particular that the support of this action includes  $\tilde{a}$. Thus 
$\out{M}{(\nu a)a}$ and $\out{M}{(\nu b)b}$ are different actions.

\begin{table*}[tb]

\begin{minipage}{1\textwidth} %{1.08\textwidth}
\begin{mathpar}

\inferrule*[Left=\textsc{In}]
    {\Psi \vdash M \sch K}
    {\framedtransempty
      {\Psi}
      {\inprefix{M}{\ve{x}}{x}.P}
      {\inlabel{K}{\ve{x}}{x}}
      {P}
    }~~
\inferrule*[left=\textsc{Out}]
    {\Psi \vdash M \sch K }
    {\framedtransempty
      {\Psi}
      {\out{M}{N}.P}
      {\out{K}{N}}
      {P}
    }~~
%\inferrule*[left={\textsc{Case}}]
%    {\Psi \vdash \varphi_i }
%    {\framedtransempty{\Psi}{\caseonly{\ci{\ve{\varphi}}{\ve{P}}}}{\tau}{P_i}}
\inferrule*[left={\textsc{Case}}]
    {\framedtransempty{\Psi}{P_i}{\alpha}{P'} \\ \Psi \vdash \varphi_i}
    {\framedtransempty{\Psi}{\caseonly{\ci{\ve{\varphi}}{\ve{P}}}}{\alpha}{P'}}

\inferrule*[Left=\textsc{Com}, Right={$\inferrule{}{\ve{a} \freshin Q
% \\\\
%\frnames{P} \freshin M\\\\ \frnames{Q} \freshin K
}$}]
 {\framedtransempty{\frass{Q} \ftimes \Psi}{P}{\bout{\ve{a}}{M}{N}}{P'} \\
  \framedtransempty{\frass{P} \ftimes \Psi}{Q}{\inlabel{K}{\ve{x}}{x}}{Q'} \\
  \Psi \ftimes \frass{P} \ftimes \frass{Q} \vdash M \sch K
  }{\framedtransempty{\Psi}{P \pll Q}{\tau}{(\nu \ve{a})(P' \pll
Q'\lsubst{N}{x})}}

\inferrule*[left=\textsc{Par},  right={$\bn{\alpha} \freshin Q$
}]
{\framedtransempty{\frass{Q} \ftimes \Psi}{P} {\alpha}{P'}}
{\framedtransempty{\Psi}{P|Q}{\alpha}{P'|Q}}

\inferrule*[left=\textsc{Scope}, right={$b \freshin \alpha,\Psi$}]
    {\framedtransempty{\Psi}{P}{\alpha}{P'}}
    {\framedtransempty{\Psi}{(\nu b)P}{\alpha}{(\nu b)P'}}

\inferrule*[left=\textsc{Open}, right={$\inferrule{}{b \freshin
\ve{a},\Psi,M\\\\
b \in \n(N)}$}]
    {\framedtransempty{\Psi}{P}{\bout{\ve{a}}{M}{N}}{P'}}
    {\framedtransempty{\Psi}{(\nu b)P}{\bout{\ve{a} \cup \{b\}}{M}{N}}{P'}}

\inferrule*[left=\textsc{Rep}]
   {\framedtransempty{\Psi}{P \pll !P}{\alpha}{P'}}
   {\framedtransempty{\Psi}{!P}{\alpha}{P'}}

\end{mathpar}

\caption{Late operational semantics. Symmetric versions of \textsc{Com}
and \textsc{Par} are elided. In the rule $\textsc{Com}$ we assume that $\fr{P} =
\framepair{\frnames{P}}{\frass{P}}$ and   $\fr{Q} =
\framepair{\frnames{Q}}{\frass{Q}}$ where $\frnames{P}$ is fresh for all of 
$\Psi, \frnames{Q}, Q, M$ and $P$, and that $\frnames{Q}$ is correspondingly
fresh. In the rule
\textsc{Par} we assume that $\fr{Q} = \framepair{\frnames{Q}}{\frass{Q}}$
where $\frnames{Q}$ is fresh for
$\Psi, P$ and $\alpha$. 
In $\textsc{Open}$ the expression $\tilde{a} \cup \{b\}$ means the sequence
$\tilde{a}$ with $b$ inserted anywhere.
}

\label{table:struct-free-labeled-operational-semantics}
\end{minipage}
\end{table*}

\begin{definition}[Transitions]
\label{transitions}

A {\em transition} is of the kind \mbox{$\framedtransempty{\Psi}{P}{\alpha}{P'}$}, meaning that when the environment contains the assertion $\Psi$ the agent $P$
can do an $\alpha$ to become $P'$.  The transitions are defined inductively in 
Table~\ref{table:struct-free-labeled-operational-semantics}.
\end{definition}
Note that $\Psi$ in Table~\ref{table:struct-free-labeled-operational-semantics}
expresses the effect that  the environment has on the agent, by enabling
conditions in \textsc{Case}, by giving rise to action subjects in \textsc{In}
and \textsc{Out} and by enabling interactions in \textsc{Com}.

Both agents and frames are identified by alpha equivalence. This means that we
can choose the bound names fresh in the premise of a rule. In a transition the
names in $\bn{\alpha}$ count as binding into both the action object and the
derivative, and transitions are identified up to alpha equivalence.
This means that the bound names can be chosen fresh, substituting each          
occurrence in both the object and the derivative. This is the reason why
$\bn{\alpha}$ is in the support of the output action: otherwise it could be
alpha-converted in the action alone. 

\begin{table*}[tb]

\begin{minipage}{1\textwidth} %{1.08\textwidth}
\begin{mathpar}

\inferrule*[left=\textsc{In}]
    {\Psi \vdash M \sch K }
{\framedtransempty{\Psi}{\inprefix{M}{\ve{y}}{x}.P}{\inn{K}{N}} {
P\lsubst { N } { x } } } \quad
\inferrule*[left=\textsc{Com}, right={$\inferrule{}{\ve{a} \freshin Q }$}]
 {\Psi \ftimes \frass{P} \ftimes \frass{Q} \vdash M \sch K \\\\
  \framedtransempty{\frass{Q} \ftimes \Psi}{P}{\bout{\ve{a}}{M}{N}}{P'} \\
  \framedtransempty{\frass{P} \ftimes \Psi}{Q}{\inn{K}{N}}{Q'}
  }
       {\framedtransempty{\Psi}{P \pll Q}{\tau}{(\nu \ve{a})(P' \pll Q')}}

\end{mathpar}
\caption{Early structured operational semantics. All other rules are as in the late semantics of Fig.~\ref{table:struct-free-labeled-operational-semantics}.}
\label{table:original-semantics}
\end{minipage}
\end{table*}

Table~\ref{table:original-semantics} gives the rules for
input and communication
of an early kind used in~\cite{bengtson.johansson.ea:psi-calculi}.
The following lemma clarifies the relation between the two semantics:
\begin{lemma}
\ 
\label{lemma:late-early}
\begin{enumerate}
\item \label{lemma:late-early:in}
$\framedtransempty{\Psi}{P}{\inn{M}N}{Q}$ in the early semantics iff there exist
$Q'$ and $x$ such that
$\framedtransempty{\Psi}{P}{\inlabel{M}{\ve{x}}{x}}{Q'}$ in the late semantics,
where $Q=Q'\lsubst{N}{x}$.
\item \label{lemma:late-early:out-tau}
For output and $\tau$ actions, $\framedtransempty{\Psi}{P}{\alpha}{Q}$ in the early semantics iff the same transition can be derived in the early semantics.
\end{enumerate}
\end{lemma}
The proof is by induction over the transition derivations. In the proof of (\ref{lemma:late-early:out-tau}), the case $\alpha=\tau$ needs both (\ref{lemma:late-early:in}) and the case where $\alpha$ is an output.

\subsection{Bisimulation}
We proceed to define early bisimulation with the late semantics:

\begin{definition}[(Early) Bisimulation]
A {\em bisimulation}
 $\cal R$ is a ternary relation between assertions and pairs of agents such that
 ${\cal R}(\Psi,P,Q)$ implies all of
 \begin{enumerate}
 \item Static equivalence:
  $\Psi \ftimes \fr{P} \sequivalent \Psi \ftimes \fr{Q}$
 \item
   Symmetry: ${\cal R}(\Psi,Q,P)$
 \item
 Extension of arbitrary assertion: $\forall \Psi'.\; {\cal
R}(\Psi \ftimes
\Psi',P,Q)$
 \item
  Simulation: for all $\alpha$, $P'$ such that $\bn{\alpha} \freshin \Psi, Q$
\begin{enumerate}
\item if $\alpha = \inlabel{M}{\ve{x}}{x}$:\quad $\Psi \frames
\trans{P}{\alpha}{P'} \Longrightarrow\\ \forall L \exists Q'\; .\; \Psi
\frames
\trans{Q}{\alpha}{Q'}$ and
$\mathcal{R}(\Psi,P'\lsubst{L}{x},Q'\lsubst{L}{x})$.
\item otherwise: \quad $\Psi \frames \trans{P}{\alpha}{P'} \Longrightarrow
\exists Q'\; .\; \Psi \frames \trans{Q}{\alpha}{Q'}$ and
$\mathcal{R}(\Psi,P',Q')$.
\end{enumerate}
\end{enumerate}
 \label{def:bisim}
We define $P \bisim Q$ to mean that there exists a bisimulation ${\cal R}$ such
that
${\cal R}(\emptyframe,P,Q)$. We also define $P \sim Q$ to mean that
$P\lsubst{L}{x} \bisim
Q\lsubst{L}{x}$ for all $x, L$.
\end{definition}

The relation between this definition and the original definition of bisimulation in~\cite{bengtson.johansson.ea:psi-calculi} is clarified by the following:

\begin{lemma}
For the psi-calculi in the present paper, a relation is a bisimulation according to Def.~\ref{def:bisim} precisely if it is a bisimulation according to \cite{bengtson.johansson.ea:psi-calculi}.
\end{lemma}
The proof is straightforward using Lemma~\ref{lemma:late-early}. As a corollary
the algebraic properties of $\sim$ established
in~\cite{bengtson.johansson.ea:psi-calculi} hold, notably that it is a
congruence.

\section{Symbolic semantics and equivalence}
\label{sec:symbolic}
The idea behind a symbolic semantics is to reduce the state space of agents. One
standard way is to avoid infinite branching in inputs by 
using a fresh name to represent whatever was received.

In psi-calculi there is an additional source of infinite branching:  a subject  
in a prefix may get rewritten
to many terms. Also here we use a fresh name to represent these terms.
This means that the symbolic actions are the same as
the concrete actions with the exception that only names are used as subjects.

A {\em symbolic transition} is of form 
\[\Psi \frames \transs{P}{\alpha}{C}{P'}\]
The intuition is that this represents a set of concrete transitions,
namely those that satisfy the constraint $C$. Before the formal
definitions we here briefly explain the rationale. Consider a
psi-calculus with integers and integer equations; for example a
condition can be ``$x=3$". An example agent is 
%$P = {\rm \bf if}\; x=3 \;{\rm \bf then}\; P'$. 
$P=\caseonly{x=3}:P'$.
If $\transs{P'}{\alpha}{\true}{P''}$, where $\true$ is a constraint that is
always true, then there should clearly be
a transition
$\transs{P}{\alpha}{C}{P''}$ for some constraint $C$ that captures
that
$x$ must be 3. One
context that can make
this constraint true is an
input, as in $a(x).P$. The input will give rise to a substitution for
$x$, and if the substitution sends $x$ to 3 the constraint is
satisfied. 
In this way the constraints are similar to those for the pi-calculus~\cite{boreale.de-nicola:symbolic-semantics,lin:symbolic-transition}.
In psi-calculi there is an additional way that a context
can enable the transition: it can contain an assertion as in
$\pass{x:=3}\,|\,P$. Concretely this agent has a
transition
$\trans{\pass{x:=3}\,|\,P}{\alpha}{\pass{x:=3}\,|\,P''}$ since $x:=3 \vdash
x=3$. 
Therefore a solution of a constraint will 
contain both a substitution of terms for names (representing the
effect of an input) and an assertion (representing the effect of a
parallel component).

\begin{definition}
The {\em atomic constraints} are of the form
$(\nu \ve{a})\constr{\Psi \vdash \varphi}$ where $\ve{a}$ are binding
occurrences into $\Psi$ and $\varphi$. A {\em solution} of an
atomic constraint is a pair $(\sigma, \Psi')$ where $\sigma$ is a
substitution of terms for names such that $\ve{a} \freshin \sigma, \Psi'$ and
$\Psi\sigma \ftimes \Psi' \vdash \varphi\sigma$. We adopt the notation $(\sigma,
\Psi) \models C$ to say that $(\sigma, \Psi)$ is a solution of $C$, and write
$\sol{C}$ for $\{(\sigma, \Psi): (\sigma, \Psi) \models C\}$.

The {\em transition constraints} are the atomic
constraints $C$ and conjunctions of atomic constraints $C \wedge C'$, where the
solutions are the intersection of the solutions for $C$ and $C'$ and we let
$(\nu \ve{a})(C \wedge C')$ mean $(\nu
\ve{a})C \wedge (\nu \ve{a})C'$.
\end{definition}
A transition constraint $C$ defines a set of solutions $\sol{C}$,
namely those where the entailment becomes true by applying the substitution and
adding the assertion.
For example, the transition constraint $\constr{\emptyframe \vdash x=3}$ has
solutions $([x:=3], \emptyframe)$ and $({\rm Id},\; x=3)$, where Id is the
identity substitution.

The structured operational symbolic
semantics is defined in Table \ref{table:symbolic-semantics}.
First consider
the {\sc Out} rule:
$\Psi \frames \transs{\out{M}{N}.P}{\out{y}{N}}{\constr{\Psi \vdash M \sch
y}}{P}$.
The symbolic subject $y$ must be chosen fresh and has a constraint
associated with it: the transition can be taken in any solution that implies
that the subject $M$ of the syntactic prefix is channel equivalent to $y$.

\begin{comment}
\TODO{Symbolic unification is only needed if we have pattern matching.}
\begin{definition}[Symbolic unification]
\label{def:symbolicunification}
$(\sigma_U, C_U)$ symbolically unifies $N$ and $\lambda\ve{x}N'$ where $\ve{x}
\freshin N$ and $\dom{\sigma_U} = \ve{x}$, written $(\sigma_U, C_U) =
\unify{N}{\lambda\ve{x}N'}$, if and only
if $\forall \sigma$ such that $\models C_U\sigma$ we have that $N\sigma =
(N'\sigma_U)\sigma$.
\end{definition}

\begin{lemma}
\label{lemma:unif1}
 If $M=N\lsubst{\ve{L}}{\ve{x}}$ where $\ve{x} \freshin M$,
then $(\lsubst{\ve{L}}{\ve{x}}, \true) = \unify{M}{\lambda\ve{x}N}$.
\end{lemma}
\begin{proof}
 We need to show that $M=N\lsubst{\ve{L}}{\ve{x}} \Rightarrow \forall \sigma
M\sigma = N\lsubst{\ve{L}}{\ve{x}}\sigma$, which is clearly true.
\end{proof}

\begin{lemma}
\label{lemma:unif2} DOES NOT HOLD!!
 If $(\sigma', \true) =
\unify{M\sigma}{\lambda\ve{x}N\sigma}$ where $\dom{\sigma'} = \ve{x}$
then $\exists (\sigma_U,C_U)$ such that $(\sigma_U)\sigma = \sigma'$,
$\models C_U\sigma$, and $(\sigma_U, C_U) = \unify{M}{\lambda\ve{x}N}$.
\end{lemma}
\end{comment}

\begin{table}[tb]

\begin{mathpar}

\inferrule*[Left=\textsc{In}, right={$y \freshin \Psi,M,P,x$}]
    { }
    {\Psi \frames
\transs{\inprefix{M}{\ve{x}}{x}.P}{\inlabel{y}{\ve{x}}{x}}
{\constr{\Psi \vdash M \sch y}}{P}}~~
\inferrule*[left={\textsc{Case}}]
    {\Psi \frames \transs{P_i}{\alpha}{C}{P'} }
    {\Psi \frames
\transs{\caseonly{\ci{\ve{\varphi}}{\ve{P}}}}{\alpha}{C \wedge \constr{\Psi
\vdash \varphi_i}}{P'}}

\inferrule*[left=\textsc{Out}, right={$y \freshin \Psi,M,N,P$}]
    { }
    {\Psi \frames \transs{\out{M}{N}.P}{\out{y}{N}}{\constr{\Psi
\vdash M \sch y}}{P}}

\inferrule*[Left=\textsc{Com}, Right={$\inferrule{}{\ve{a} \freshin Q,
%\textcolor{red}{\ve{x} \freshin P}
\\\\
y \freshin z\\\\
\Psi' = \Psi \ftimes \frass{P} \ftimes \frass{Q}\\\\
%\dom{\sigma} \subseteq \ve{x}\\\\
%\ran{\sigma} \subseteq \subterm{N}\\\\
%\ve{a} \subseteq \names{\ran{\sigma}}
}$}]
    {\Psi \ftimes \frass{Q} \frames
\transs{P}{\sbout{y}{\ve{a}}{N}}{(\nu \ve{b_P})\constr{\Psi' \vdash M_P \sch
y} \wedge C_P}{P'} \\
     \Psi \ftimes \frass{P} \frames
\transs{Q}{\inlabel{z}{\ve{x}}{x}}{(\nu \ve{b_Q})\constr{\Psi' \vdash M_Q \sch
z}
\wedge C_Q}{Q'}
    }
    {\Psi \frames \transs{P \pll Q}{\tau}{(\nu
\ve{b_P},\ve{b_Q})\constr{\Psi'
\vdash M_P \sch M_Q} \wedge C_P \wedge C_Q}{(\nu
\ve{a})(P' \pll Q'\lsubst{N}{x})}}

\inferrule*[Left=\textsc{Par}, right={$\inferrule{}{\bn{\alpha} \freshin Q
\\\\ \alpha=\tau \vee \subj{\alpha} \freshin Q}$}]
    {\Psi \ftimes \frass{Q} \frames \transs{P}{\alpha}{C}{P'}}
    {\Psi \frames \transs{P \pll Q}{\alpha}{(\nu \frnames{Q})C}{P' \pll Q}}

\inferrule*[Left=\textsc{Scope}, right={$a \freshin \alpha, \Psi$}]
    {\Psi \frames \transs{P}{\alpha}{C}{P'}}
    {\Psi \frames \transs{(\nu a)P}{\alpha}{(\nu a)C}{(\nu a)P'}}

\inferrule*[left=\textsc{Open}, right={$\inferrule{}{a \in \n(N) \\\\ a
\freshin \ve{a},\Psi,y}$} ]
    {\Psi \frames \transs{P}{\sbout{y}{\ve{a}}{N}}{C}{P'}}
    {\Psi \frames \transs{(\nu a)P}{\sbout{y}{\ve{a} \cup a}{N}}{(\nu a)C
%\wedge \constr{a \freshin (\nu a)P}
}{P'}}

\inferrule*[left=\textsc{Rep}]
   {\Psi \frames \transs{P \pll !P}{\alpha}{C}{P'}}
   {\Psi \frames \transs{!P}{\alpha}{C}{P'}}

\end{mathpar}
\caption{Transition rules for the symbolic semantics. Symmetric versions of
\textsc{Com} and \textsc{Par} are
elided. 
In the rule $\textsc{Com}$ we assume
that $\fr{P} = \framepair{\frnames{P}}{\frass{P}}$ and   $\fr{Q} =
\framepair{\frnames{Q}}{\frass{Q}}$ where $\frnames{P}$ is
fresh for all of $\Psi, \frnames{Q}, Q$ and $P$, and that $\frnames{Q}$ is
correspondingly
fresh. We also assume that $y,z \freshin
\Psi,\frnames{P},P,\frnames{Q},Q,N, \ve{b_P}, \ve{b_Q},\ve{a}$.
In
the rule \textsc{Par} we assume that $\fr{Q} =
\framepair{\frnames{Q}}{\frass{Q}}$
where $\frnames{Q}$ is
fresh for $\Psi, P$ and $\alpha$. 
In $\textsc{Open}$ the expression $\tilde{a} \cup \{a\}$ means the sequence
$\tilde{a}$ with $a$ inserted anywhere.}
\label{table:symbolic-semantics}
\end{table}

The rule \textsc{Com} is of particular interest. The intuition is that the
symbolic action subjects are placeholders for the values $M_P$ and $M_Q$.
In the conclusion the constraint is that these are channel
equivalent, while $y$ and $z$ will not occur again.

We will often write $P\sgto{\alpha}{C} P'$ for $\emptyframe\frames P\sgto{\alpha}{C} P'$.

\subsection{Symbolic bisimulation}

In order to define a symbolic bisimulation we need additional kinds of
constraints. If a process $P$ does a bound output
$\bout{\ve{a}}{y}{N}$ that is matched by a bound output $\bout{\ve{a}}{y}{N'}$
from $Q$ we need constraints that keep track of the fact that $N$ and $N'$
should be syntactically the same, and that $\ve{a}$ is sufficiently fresh.
\begin{definition}
The {\em constraints} include the transition constraints, 
 $\constr{M=N}$, and
$\constr{a \freshin X}$, where $X$ is any nominal data type. The solutions of
the last two are all pairs $(\Psi, \sigma)$ such that $M\sigma = N\sigma$
and $a \freshin (X\sigma)$ respectively. We also include conjunction of
constraints $C \wedge C'$, where the set of solutions is the
intersection of the
solutions for $C$ and $C'$.
\end{definition}
Note that the assertion part of the solution is irrelevant for constraints of kind $\constr{M=N}$ and
$\constr{a \freshin X}$, and that the substitution does not affect $a$ in $\constr{ a \freshin X}$.
The constraint $\constr{M=N}$ is
used in the bisimulation for matching output objects, and $\constr{a \freshin
X}$ is
used in the bisimulation for recording what an opened name must be fresh for.
This corresponds to distinctions in open bisimulation for the pi-calculus
\cite{sangiorgi:theory-bisimulation}.
We define $\true$ to be $\constr{M=M}$, we write $\constr{a \freshin X,Y}$ for
$\constr{a \freshin X}
\wedge \constr{a \freshin Y}$, and we extend the notation to sets of names, e.g.
$\constr{\ve{a}
\freshin X}$.

\begin{definition}[Constraint implication]
 A constraint $C$ \emph{implies} another constraint $D$, written $C \cimplies
D$, iff $\sol{C} \subseteq \sol{D}$.
%
%We extend constraint implication to disjunctions:
We write $C \cimplies \bigvee \ve{C}$ iff
 for
each $(\sigma,\Psi) \in \sol{C}$ there exists a $C' \in \ve{C}$ such that $(\sigma,\Psi) \in
\sol{C'}$.
\end{definition}

Before we can give the definition of symbolic
bisimulation we need to define a symbolic variant of  the concrete static equivalence.

\begin{definition}[Symbolic static equivalence]
 \label{def:ssequivalent}
 Two processes $P$ and $Q$ are \emph{statically equivalent} for $C$, written
$P \sequivalent_C Q$, if for each $(\sigma,\Psi) \in \sol{C}$
we have that $\Psi \ftimes \fr{P}\sigma \sequivalent \Psi \ftimes
\fr{Q}\sigma$.
\end{definition}

\begin{comment}
\begin{definition}[Equal actions]
 Two actions $\alpha_s$ and $\beta_s$ are equal for $C$, written
$\alpha_s \lequal{C} \beta_s$, is defined as
\[
\begin{array}{rcl}
 (\sigma, \Psi) \models C \wedge \ve{a} \freshin \sigma,\Psi \Longrightarrow
N\sigma = N'\sigma
&\Leftrightarrow& \sbout{y}{\ve{a}}{N} \lequal{C} \sbout{y}{\ve{a}}{N'}\\
\true &\Leftrightarrow& \lin{y}{x}{x} \lequal{C} \lin{y}{x}{x}\\
\true &\Leftrightarrow& \tau \lequal{C} \tau
\end{array}
\]
\end{definition}
\end{comment}

We now have everything we need to define symbolic bisimulation. This definition
follows the definition in \cite{hennessy.lin:symbolic-bisimulations} closely.

\begin{definition}[(Early) Symbolic bisimulation]\label{def:symbisim}
A {\em symbolic bisimulation}
 $\cal S$ is a ternary relation between constraints and pairs of agents such
that ${\cal S}(C,P,Q)$ implies all of

\begin{enumerate}
\item $P \sequivalent_C Q$, and
\item ${\cal S}(C,Q,P)$, and
\item If $\transs{P}{\alpha}{C_P}{P'}$,
$\bn{\alpha} \freshin (P,Q,C,C_P,\subj{\alpha})$ and $\subj{\alpha}
\freshin (P,Q,C)$ then there exists a set of constraints
$\widehat{C}$ such that 
%$\bn{\alpha} \freshin \widehat{C}$, 
$C \wedge C_P \cimplies \bigvee \widehat{C}$ \\
and for all
$C' \in \widehat{C}$ there exists $Q'$, $\alpha'$, and $C_Q$ such that

\begin{enumerate}
\item $\transs{Q}{\alpha'}{C_Q}{Q'}$, and
\item $C' \cimplies C_Q$, and
\item if $\alpha = \bout{\ve{a}}{y}{N}$ then $\alpha' =
\bout{\ve{a}}{y}{N'}$, $C' \cimplies \constr{N=N'}$, \\
$\quad$ and $(C' \wedge
\constr{\ve{a} \freshin P,Q}, P', Q') \in {\cal S}$

otherwise $\alpha = \alpha'$ and $(C', P', Q') \in {\cal S}$
\end{enumerate}
\end{enumerate}
We write $P\sim_s Q$ if $(\true, P,Q)\in\mathcal{S}$ for some symbolic
bisimulation $\mathcal{S}$, and say that $P$ is \emph{symbolically
bisimilar} to $Q$.
\end{definition}

The set $\widehat{C}$ allows a case analysis on the constraint solutions, as examplified in the next section. 
The output objects need to be equal in a solution to $C'$. Since
the solutions of  $\constr{N=N'}$ only depend on the substitutions, this constraint corresponds to the fact that the objects must be identical in the concrete bisimulation.
Note that $\bn{\alpha}$ may occur in $\widehat{C}$. Based on
\cite{boreale.de-nicola:symbolic-semantics,lin:symbolic-transition}, we
conjecture that adding the requirement
$\bn{\alpha} \freshin \widehat{C}$ would give late symbolic bisimulation.

\section{Examples}
\label{sec:examples}
We now look at a few examples to illustrate the concrete and
symbolic transitions and bisimulations. First consider a simple example from the pi-calculus.
This can be expressed as a psi-calculus: let the only data terms be
names,  the only
assertion be $\emptyframe$, the conditions be equality and inequality tests on
names, and entailment defined by $\forall a.\emptyframe\vdash a=a$, $\forall a,b:a\neq b.\emptyframe\vdash a\neq b$ and $\forall
a.\emptyframe\vdash a\sch a$. For a more thorough discussion,
see \cite{bengtson.johansson.ea:psi-calculi}. %We use $\tau.P$ as a
%shorthand for $(\nu b)(\overline{b}b.\nil\parop\underline{b}(b).P)$ for
%some $b\freshin P$.
In the following examples we drop a trailing $.\nil$.
Consider the two agents $P_1$ and $Q_1$:
\[
\begin{array}{rl@{\quad\text{where }}rcl}
 P_1 &= a(x)\sdot P_1' & P_1'&=&\out{a}{b}\sdot \out{a}{b}\\
 Q_1 &= a(x)\sdot Q_1' &
Q_1'&=&(\caseonly{\ci{x=b}{\out{a}{b}\sdot \out{a}{b}}\;\casesep\;\ci{x\neq
b}{\out{a}{b} \sdot \out{a}{b}}})
\end{array}
\]
These are bisimilar. A concrete bisimulation between these agents is
\[
 \{ (\emptyframe,P_1,Q_1) \} \cup \bigcup_{n \in \N} \{ (\emptyframe,
P_1', Q_1'\lsubst{n}{x} \}\: \cup \{ (\emptyframe,\out{a}{b},\out{a}{b}) \}
\]
The bisimulation needs to be infinite because of the infinite branching in the input.
In contrast, a
symbolic bisimulation only contains four triples:
\[
\left\{
\begin{array}{@{}c@{\;\;}c@{\;\;}c@{\;\;}c@{}}
(\true,P_1,Q_1), & (\true, P_1', Q_1'), &
(\constr{\emptyframe \vdash x=b},\out{a}{b},\out{a}{b}), & (\constr{\emptyframe \vdash x\neq b},\out{a}{b},\out{a}{b})\\
\end{array}
\right\}
\]

When checking the second triple $(\true,P_1',Q_1')$, the transition of
$P_1'$ is matched by a case analysis: $\widehat{C}$ in the definition of symbolic bisimulation (Def.~\ref{def:symbisim}) is
$\{\constr{\emptyframe \vdash x=b}, \constr{\emptyframe \vdash x\neq
b}\}$, and a matching transition for $Q_1'$ can be found for each of
these cases, so the agents are bisimilar.
In contrast,
they are not equivalent in the incomplete
symbolic bisimulations in
\cite{borgstroem.briais.ea:symbolic-bisimulation} and
\cite{DBLP:conf/fsttcs/DelauneKR07}.

Next we look at an example where we have tuples of channels and projection,
e.g. the entailment relation gives us that $\emptyframe \vdash \sym{first}{M,N}
\sch M$. Consider the agent
\[
 R = \out{M}{N}\sdot R'
\]
Concretely this agent has infinitely many transitions even in an empty frame: $\trans{R}{\out{M}{N}}{R'}$, 
\iftrue %false
and equivalent actions $\out{\sym{first}{M,K}}{N}$ 
for all $K$, and
$\out{\sym{first}{\sym{first}{M,L},K}}{N}$ 
for all $L$ and $K$, etc.
\else
and equivalently $\trans{R}{\out{\sym{first}{M,K}}{N}}{R'}$ 
for all $K$, and
$\trans{R}{\out{\sym{first}{\sym{first}{M,L},K}}{N}}{R'}$ 
for all $L$ and $K$, etc.
\fi
Symbolically, however, it has
only one transition: $\transs{R}{\out{y}{N}}{\constr{\emptyframe \vdash M \sch y}}{R'}$.

For another example, consider the two agents
\[
\begin{array}{rl@{\qquad}rl}
 P_2 &= \out{F}N\sdot P' &%\\
 Q_2 &= \nil
\end{array} 
\]
where $F$ is a term such that for no $\Psi,M$ does it hold that $\Psi \vdash F \sch M$, i.e., $F$ is not a channel. Then we have that $P_2$ and $Q_2$ are concretely bisimilar since
neither one of them has a transition. But symbolically $P_2$ has the
transition $\transs{P_2}{\out{y}{N}}{\constr{\emptyframe \vdash F \sch y}}{P'}$, while
$Q_2$ has no symbolic transition. Perhaps surprisingly they are still symbolically
bisimilar: Def.~\ref{def:symbisim} requires that we find a
disjunction $\widehat{C}$ such that $C \wedge C_P \cimplies \bigvee\widehat{C}$, or in
this case such that $\true \wedge \constr{\emptyframe \vdash F \sch y}
\cimplies \bigvee\widehat{C}$. Since $F$ is not channel equivalent to
anything, the left hand side has no solutions, which means that any set
$\widehat{C}$ will do, and in particular the empty one.
The condition ``for all $C'\in\widehat{C}\,$" in the definition becomes
trivially true, so $Q_2$ does not have to mimic the transition.

A final example shows the use of cryptographic primitives. Here
 the terms contains $\sym{enc}{M,k}$ and $\sym{dec}{M,k}$,
assertions are variable assignments, e.g. $x := M$, the conditions are
equality tests between terms, and the entailment relation is parametrised by an
equation system which contains the equation $\sym{dec}{\sym{enc}{M,k},k}=M$. Consider
%and variable assignments are applied to the condition. Consider the agents
\[
\begin{array}{ll}
 P_3 &= (\nu a,k)\,(\pass{x := \sym{enc}{a,k}} \,\pll\,
b(z)\,.\,\out{b}{k}\,.\, (\caseonly{\ci{z=a}{\out{c}{d}}}))\\
 Q_3 &= (\nu a,k)\,(\pass{x := \sym{enc}{a,k}} \,\pll\,
b(z)\,.\,\out{b}{k})
\end{array}
\]
Here the environment can use $x$, the result of encrypting $a$ with $k$, but not the bound $a$ or $k$.
Intuitively these agents are bisimilar since the key $k$ is not revealed until
after the agents receive $z$, which therefore cannot be equal to $a$. The first
symbolic transitions of the agents are
\[
 \begin{array}{ll}
  \transs{P_3}{y(z)}{(\nu a,k)\constr{\emptyframe \vdash b
\sch y}}{(\nu a,k)(\pass{x := \sym{enc}{a,k}} \,\pll\,
\out{b}{k}\,.\,(\caseonly{\ci{z=a}{\out{c}{d}}}))} & = P_3'\\
   \transs{Q_3}{y(z)}{(\nu a,k)\constr{\emptyframe \vdash b
\sch y}}{(\nu a,k)(\pass{x := \sym{enc}{a,k}} \,\pll\, \out{b}{k})} & = Q_3'
 \end{array}
\]
and the second transitions are
\[
 \begin{array}{ll}
   \transs{P_3'}{\out{y'}{(\nu k)k}}{(\nu a,k)\constr{\emptyframe
\vdash b \sch
y'}}{(\nu a)(\pass{x := \sym{enc}{a,k}} \,\pll\,
(\caseonly{\ci{z=a}{\out{c}{d}}}))} & =
P_3''\\
   \transs{Q_3'}{\out{y'}{(\nu k)k}}{(\nu a,k)\constr{\emptyframe
\vdash b \sch
y'}}{(\nu a)(\pass{x := \sym{enc}{a,k}})} & = Q_3''
 \end{array}
\]
A symbolic
bisimulation, where we for simplicity ignore the constraints that arise for
subjects, is
\[
\{(\true,P_3, Q_3),\quad (\true, P_3', Q_3'),\quad (\constr{k\freshin P_3',Q_3'}, P_3'', Q_3'')\}
\]
Here the constraint $\constr{k \freshin P_3',Q_3'}$ will among other things imply
that $k \freshin z$. The final transition of $P_3''$
has the constraint
$(\nu a)\constr{\emptyframe \vdash z=a}$, so 
we must find a
disjunction $\widehat{C}$ such that
$k \freshin P_3',Q_3' \wedge (\nu a)\constr{\emptyframe \vdash z=a} \cimplies \widehat{C}$.
Since $a$ is bound, the only way to find a solution to
the left hand side is to find a value for $z$ that evaluates to $a$. One
candidate for a solution is $(\lsubst{\sym{dec}{x,k}}{z}, \emptyframe)$, but
because of the constraint $k \freshin z$ this does not work. In fact, 
there is no
solution to the left hand side because of the freshness
constraint on $k$ and the fact that $a$ is bound. This means that, as in the previous example, any
disjunction
$\widehat{C}$ will do, and in particular the empty disjunction, and trivially
$Q_3''$ does not have to mimic the transition.

In contrast, if we swap the order of the inputs and the outputs in $P_3$
and $Q_3$ and try to construct the bisimulation relation we will discover that
we do not get the constraint $k \freshin z$. This means that 
$(\lsubst{\sym{dec}{x,k}}{z}, \emptyframe)$ is a solution
%we have solutions
to $C \wedge C_P$ in the definition of bisimulation, and that $Q_3''$ must
mimic the transition from $P_3''$. In this case the agents are not bisimilar.

\section{Results}
\label{sec:results}
We now turn to showing that the concrete and symbolic equivalences coincide.
% \begin{definition}[Substitution on actions, extension of substitution]

% \noindent
We define substitution on symbolic actions by
$\tau\sigma = \tau$,
$(\inlabel{y}{\ve{x}}{x})\sigma = \inlabel{y\sigma}{\ve{x}}{x\sigma}$, and
$(\bout{\ve{a}}{y}{N})\sigma = \bout{\ve{a}}{y\sigma}{N\sigma}$,
where $x,\ve{a} \freshin \sigma$.
We define the substitution $\sigma \cdot \lsubst{M}{y}$ for $y \freshin \sigma$
by
$(\sigma\cdot\lsubst{M}{y})(x) = M$ if $x=y$, and $\sigma(x)$ otherwise.
%For $y \freshin \sigma$, we define $\sigma \cdot \lsubst{M}{y}(x)$ as $M$ if $x=y$, and $\sigma(x)$ otherwise.
%\end{definition}

The following two lemmas show the operational correspondence between the
symbolic semantics and the concrete semantics: given a symbolic transition
where the transition constraint has a solution, there is
always a corresponding concrete transition
(Lemma~\ref{lemma:correspondence-symbolic-concrete}) and vice versa
(Lemma~\ref{lemma:correspondence-concrete-symbolic}).

\begin{lemma}[Correspondence symbolic-concrete]\mbox{}\\
\label{lemma:correspondence-symbolic-concrete}
\vspace{-1\baselineskip}
\begin{enumerate}
 \item If $\transs{P}{\inlabel{y}{\ve{x}}{x}}{C}{P'}$ then for all $(\sigma,
\Psi) \in
\sol{C}$ s.t. $x \freshin \sigma$
we have that $\Psi \frames
\trans{P\sigma}{(\inlabel{y}{\ve{x}}{x})\sigma}{P'\sigma}$.
%\vspace{0.5\baselineskip}
 \item If $\transs{P}{\bout{\ve{a}}{y}{N}}{C}{P'}$ then for all $(\sigma, \Psi)
\in \sol{C}$ s.t. $\ve{a} \freshin
\sigma$ we have that $\Psi \frames
\trans{P\sigma}{(\bout{\ve{a}}{y}{N})\sigma}{P'\sigma}$.
%\vspace{0.5\baselineskip}
\item If $\transs{P}{\tau}{C}{P'}$ then for all $(\sigma,
\Psi) \in \sol{C}$ we have that $\Psi
\frames
\trans{P\sigma}{\tau}{P'\sigma}$.
\end{enumerate}
\end{lemma}

\begin{lemma}[Correspondence concrete-symbolic]\mbox{}\\
\label{lemma:correspondence-concrete-symbolic}
\vspace{-1\baselineskip}
\begin{enumerate}
 \item   If $\Psi \frames
\trans{P\sigma}{\inlabel{M}{\ve{x}}{x}}{P'\sigma}$,\quad $y \freshin
P,\sigma,M,N,x$,\quad
and $x \freshin \sigma, P$
then there exists $\ve{b}, M_P$, and $C_P$  such that
\\$\transs{P}{\inlabel{y}{\ve{x}}{x}}{(\nu
\ve{b})\constr{\frass{P} \vdash M_P \sch y}\wedge C_P}{P'}$
and $(\sigma\cdot\lsubst{M}{y},\Psi) \in \sol{(\nu
\ve{b})\constr{\frass{P}
\vdash M_P \sch
y} \wedge C_P}$.
%\vspace{0.5\baselineskip}
 \item   If $\Psi \frames
\trans{P\sigma}{\bout{\ve{a}}{M}{N\sigma}}{P'\sigma}$,\quad $y \freshin
P,\sigma,M,\ve{a}$,\quad and $\ve{a} \freshin \sigma, P$
then there exists $\ve{b}, M_P$, $C_P$ such that
\\$\transs{P}{\bout{\ve{a}}{y}{N}}{(\nu
\ve{b})\constr{\frass{P} \vdash M_P \sch y} \wedge C_P}{P'}$
 and
$(\sigma\cdot\lsubst{M}{y},\Psi) \in \sol{(\nu
\ve{b})\constr{\frass{P}
\vdash M_P \sch y} \wedge C_P}$.
%\vspace{0.5\baselineskip}
 \item   If $\Psi \frames \trans{P\sigma}{\tau}{P'\sigma}$
then there exists $C$ such that $\transs{P}{\tau}{C}{P'}$
 and $(\sigma,\Psi) \in \sol{C}$.
\end{enumerate}
We assume in 1 and 2 that $\fr{P} = \frameng{\frnames{P}}{\frass{P}}$ and
$\frnames{P},\ve{b} \freshin y,\Psi,\sigma,P$.
\end{lemma}

The proofs are by induction over the transition derivation (one case for each
rule).%; for the details see the Appendix.

\begin{theorem}[Soundness]
\label{theorem:soundness}
Assume $\mathcal{S}$ is a symbolic bisimulation and let \\ $\mathcal{R} = \{(\Psi,
P\sigma, Q\sigma) : \exists C . (\sigma,\Psi)
\models C \text{ and } (C,P,Q) \in \mathcal{S}\}$.
Then $\mathcal{R}$ is a concrete bisimulation.
\end{theorem}

The proof idea to show that $\mathcal{R}$ is a concrete bisimulation is to
assume $(\Psi, P\sigma, Q\sigma) \in \mathcal{R}$ and that $P\sigma$ has a
transition in environment $\Psi$. We use
Lemma~\ref{lemma:correspondence-concrete-symbolic} to find a symbolic transition
from $P$, then the fact that $\mathcal{S}$ is a symbolic bisimulation to find a
simulating symbolic transition from $Q$, and finally
Lemma~\ref{lemma:correspondence-symbolic-concrete} to find the required concrete
transitions from $Q\sigma$. %For details see the Appendix.

Similarly to~\cite{hennessy.lin:symbolic-bisimulations} we need an extra assumption about the expressiveness of constraints: for all ${\cal R}, P, Q$ such that
$\cal R$ is a concrete bisimulation there exists a constraint $C$ such that
\((\Psi, \sigma) \models C \Longleftrightarrow (\Psi, P\sigma, Q\sigma) \in {\cal R}\).
%The assumption can be trivially satisfied by extending the set of constraints
%to include all such $C$. This will ensure that our results hold but is
%unsatisfactory from a practical point of view.
In order to determine symbolic bisimilarity in an efficient way we
need to compute this constraint, which is easy for the
pi-calculus~\cite{boreale.de-nicola:symbolic-semantics,lin:symbolic-transition,lin:computing-bisimulations}
and harder (but in many practical cases possible) for cryptographic
signatures~\cite{borgstroem:equivalences-calculi}. These results suggest that
our constraints are sufficiently expressive, but for other instances of
psi-calculi we may have to extend the constraint language. We leave this as an
area of further research.

\begin{theorem}[Completeness]
\label{theorem:completeness}
Assume that $\mathcal{R}$ is a concrete bisimulation and let \\ $\mathcal{S} =
\{(C,P,Q) : (\sigma,\Psi) \models C \text{ implies }
(\Psi, P\sigma, Q\sigma) \in \mathcal{R}\}$. Then $\mathcal{S}$ is a symbolic
bisimulation.
\end{theorem}

The proof idea is the converse of the proof for Theorem~\ref{theorem:soundness}. The expressiveness assumption of constraints mentioned above is needed in order to construct the disjunction of constraints in the symbolic bisimulation. 
%For details see the Appendix.
%
From these two theorems we get:
\begin{corollary}[Full abstraction]
$P \sim Q$ if and only if $P \sim_s Q$.
\end{corollary}

\section{Conclusion and Future Work}

\label{sec:conclusion}
We have defined a symbolic operational semantics for psi-calculi and a
symbolic bisimulation which is fully abstract wrt the original
semantics. 
While the developments
in~\cite{bengtson.johansson.ea:psi-calculi} give meta-theory for a
wide range of calculi of mobile processes with nominal data and
logic, the work presented in this paper gives a solid foundation for
automated tools for the analysis of such calculi.  
%avoid ``all such calculi''

%\TODO{a bit of technical challenges and design decisions}
As mentioned in the introduction, the purity of the original semantics
of psi-calculi has made the symbolic semantics easier to develop.
There are no structural equivalence rules (which are a complication in
\api), the scope opening rule is because of this straight-forward
which makes knowledge representation simpler than in spi-calculi, and
the bisimulation less complex.
Nevertheless, the technical challenges have not been absent: the
precise design of the constraints and their solution has been
delicate.  Since assertions may occur under a prefix, the environment can
change after a transition. Keeping the assertion $\Psi$ in the
transition constraints
(on the form $(\nu \ve{a})\constr{\Psi \vdash \varphi}$) essentially keeps
a snapshot of the environment that gives rise to the transition. An
alternative would be to use time stamps to keep track of which environment
made which condition true, but that approach seems more difficult.

%\TODO{weak vs strong bisimulation}
Our symbolic bisimulation is a strong equivalence which does not
abstract the internal $\tau$ transitions. This is less useful for
verification than a weak observational equivalence, but still a
significant step towards mechanized verification.
We are currently developing a weak bisimulation for psi-calculi, and
are studying the correspondence to a barbed bisimulation congruence.
Preliminary results indicate that lifting the symbolic bisimulation
presented here to weak bisimulation will be unproblematic.

The original psi-calculi admit pattern matching in inputs. 
In a symbolic semantics this would lead to complications in the \textsc{COM}-rule,
which should introduce a substitution
for the names bound in the pattern. This means introducing
more fresh names and constraints, and it is not clear that the
convenience of pattern matching outweighs such an awkward semantic
rule. We leave this as an area for further study.

For future work, we need to develop an algorithm for deciding symbolic
bisimulation and implement it in a tool. A natural basis for this
would be the algorithm given in
\cite{hennessy.lin:symbolic-bisimulations}.
Furthermore, the termination of the algorithm will depend
on the properties of the parameters of the particular psi-calculus: it
is easy to construct a psi-calculus where the entailment relation or
static equivalence is
not decidable, but in many practical cases it will be~\cite{borgstroem:equivalences-calculi,baudet:thesis}.  
We intend to use constraint solvers developed for specific application
domains (e.g. security) in a future generic tool. We will also produce mechanized proofs of the adequacy of the symbolic
semantics, using the Isabelle theorem prover.

When typing schemes have been developed for psi-calculi,
a natural progression would be to take advantage of those also in the
symbolic semantics, to further constrain the possible values and thus the
size of state spaces.

\bibliographystyle{eptcs}
\bibliography{sos09-symbolic-psi,pi}

%\newpage
%\appendix
%\input{proofs}
\end{document}